\documentclass[10pt]{article}

\usepackage[letterpaper,margin=0.9in]{geometry}
\usepackage[T1]{fontenc}
\usepackage[utf8]{inputenc}
\usepackage{lmodern}
\usepackage{microtype}
\usepackage{amsmath,amssymb,amsthm,mathtools,bm}
\usepackage{tikz}
\usetikzlibrary{arrows.meta,positioning,calc,shapes.geometric,patterns}
\usepackage{xcolor}
\usepackage{hyperref}
\usepackage[nameinlink,capitalise,noabbrev]{cleveref}
\usepackage{authblk}
\hypersetup{
  colorlinks=true,
  linkcolor=blue!55!black,
  citecolor=blue!55!black,
  urlcolor=blue!55!black,}

\numberwithin{equation}{section}
\newtheorem{theorem}{Theorem}[section]
\newtheorem{proposition}[theorem]{Proposition}

\newcommand{\A}{\mathcal A}
\newcommand{\Hh}{\mathcal H}
\newcommand{\R}{\mathbb R}
\newcommand{\id}{\mathrm{id}}

\newcommand{\dd}{\mathrm d}
\newcommand{\ii}{\mathrm i}
\newcommand{\e}{\mathrm e}
\newcommand{\Tr}{\operatorname{Tr}}

\title{An Algebraic Formulation of Thermality in Nested Rindler Wedges}

\author[1]{Waheed A. Dar}
\author[1]{Prince A. Ganai}
\author[2]{Nirmalya Kajuri}

\affil[1]{Department of Physics, National Institute of Technology Srinagar,
Kashmir-190006, India\\}

\affil[2]{School of Physical Sciences, Indian Institute of Technology Mandi,
Himachal Pradesh, India\\
\vspace{2mm}\\
\texttt{Email: waheed.dar1729@gmail.com, princeganai@nitsri.ac.in, nirmalya@iitmandi.ac.in}}

\date{\today}

\begin{document}
\maketitle
\vspace{-2.3em}

\begin{abstract}
 Padmanabhan and Kolekar had demonstrated the phenomenon of recursive thermality in two dimensional free field theories where the Rindler vacuum appears thermal to an observer accelerated with respect to the Rindler observer. We study the same construction at an algebraic level for conformal field theories.
 At each step, a parent wedge is conformally mapped to a reference Minkowski spacetime, where a new right Rindler wedge and its left sibling are selected. We illustrate that under the Bisognano--Wichmann property, thermality holds at every generation: the state associated with a parent wedge becomes a KMS state when restricted to its right child. Further, using the modular structure of the wedge and its causal sibling, we identify the sibling algebra with the modular commutant and derive the corresponding recursive modular flow explicitly. 
\end{abstract}

\tableofcontents

\section{Introduction}
\label{sec:introduction}

The Unruh effect shows that the Minkowski vacuum restricted to a
Rindler wedge is thermal with respect to the one-parameter Lorentz-boost flow preserving that wedge
\cite{Fulling:1972md,Davies:1974th,Unruh:1976db}. The algebraic formulation of the same is given by the Bisognano--Wichmann theorem \cite{Bisognano:1975ih,Bisognano:1976za,Schroer:1998pj,Doplicher:1984zz}. The Bisognano--Wichmann theorem shows that in the vacuum representation, the
modular automorphism group of the wedge algebra is generated by Lorentz
boosts. The restricted vacuum is therefore a KMS state for the
wedge-preserving boost flow. Physically, this corresponds to the Minkowski vacuum appearing as a thermal bath to the accelerated observer.

Now suppose there is a second observer, accelerated with respect to the original accelerated observer. How would they observe the vacuum of the first observer i.e. the Rindler vacuum? The question was answered for a massless free field theory in the context of 2 dimensions by Kolekar and Padmanabhan in \cite{Kolekar:2013hra}. They applied a
second Rindler transformation inside the original Rindler wedge and showed that the Rindler vacuum appears
thermal in the second  frame (which they called ``Rindler--Rindler''). They also found that the spectrum of
particles seen by Rindler--Rindler observers in the Minkowski vacuum is
identical to that seen by an observer accelerating through a real thermal
bath \cite{Kolekar:2013aka}. The construction was later extended to a sequence
of nested Rindler frames and it was shown that the vacuum of an earlier Rindler
frame appears thermal in a later one
\cite{Lochan:2025mru}. Other recent works that have studied the nested frames in 2 dimensions are
\cite{Nair:2024psn,Sahota:2024thk,Dubey:2025hab,Jha:2026qrp,Derezinski:2004xz,Gutti:2022xov}.

These studies describe
recursive thermality mainly through particle spectra, Bogoliubov
transformations, and detector responses \cite{Lochan:2016cxt,Jha:2026prm}. In this paper, we ask whether the same
recursive construction can be formulated directly at the level of states,
observables, and the flows that preserve the corresponding regions. This
leads naturally to a KMS and modular description of recursive Rindler
thermality. We use slightly different terminology, we call the original wedge the parent region and its Rindler transform the child region.  

Our starting point is the geometric intuition behind the effect described in \cite{Kolekar:2013hra,Lochan:2025mru}.The key point here is that in 2-dimensions, the map between Minkowski space and Rindler space is conformal. For a free field theory, the conformally transported Minkowski vacuum is the Rindler vacuum. The map from the parent Rindler wedge to the child wedge (or from Rindler to Rindler-Rindler) is again the exact same map as from Minkowski to Rindler. Thus the same relation repeats: the restricted parent vacuum becomes thermal in the child wedge. This relation repeats for every subsequent child wedge ,see the figure \ref{fig:nested-rindler}

In more detail, 
each Rindler transformation takes a
parent region and introduces two causally complementary children
associated with the right and left wedges of a reference Minkowski plane.
The reference Minkowski vacuum provides the state from which the corresponding
parent state is constructed. The right child is used as the parent for the
next step, while the same construction can be applied to every sibling to
generate a binary causal tree. We focus on the right-descending branch
\begin{equation}
  W_0\supset W_1\supset W_2\supset\cdots .
  \label{eq:intro-chain}
\end{equation}

This also explains why the recursive thermality phenomenon is specific to
two dimensions. In higher dimensions, the mapping between Minkowski space
and Rindler space is not conformal. Two dimensions is special because the
local conformal maps factor into independent reparametrizations of the
two null coordinates $u$ and $v$. This allows each parent to be mapped back
to a complete reference plane and the construction to be repeated,
producing the closed iterated-log structure used below. The construction
also requires the underlying theory to be conformal. This is why the
original construction was formulated for massless free field theories;
here we extend it to conformal field theories.

The central result is that the state associated with each parent region is thermal when restricted to its right child, in the precise KMS sense:
\begin{equation}
  \left.\omega_n\right|_{\A(W_{n+1})}
  \text{ is KMS for }\alpha_s^{(n+1)}
  \quad\text{with}\quad
  \beta_{n+1}=\frac{2\pi}{g_{n+1}}.
  \label{eq:intro-main}
\end{equation}
Thus thermality at successive Rindler levels follows from the combination of conformal transport and the Bisognano--Wichmann property. It is established directly at the level of states and local algebras, rather than being inferred only from particle spectra or detector responses.

When the reference vacuum representation is irreducible, the transported
parent state is pure, while its restriction to the child is KMS. With wedge duality, the causally complementary sibling is identified with the modular commutant:
\begin{equation}
  J_n M_{\omega_n}(W_{n+1})J_n
  =M_{\omega_n}(W_{n+1})'
  =M_{\omega_n}(D_n).
  \label{eq:intro-commutant}
\end{equation}
Thus the pure parent gives a thermal child, while the causally complementary sibling realizes the corresponding modular commutant.

A second result follows by shifting the root of the recursive construction by one generation. The state $\omega_{n-1}$, when restricted to $W_n$, pulls back to a thermal state $\tau_n$ on the reference Minkowski plane. The next recursive experiment is then an ordinary Rindler split of a thermal state. For $n=1$ this gives a state--algebra--flow formulation of the second observation of Ref.~\cite{Kolekar:2013hra}. The ordinary Rindler observer flow in the thermal state is not, in general, the modular flow of the restricted thermal state. The latter is the geometric thermal-sector modular flow studied in Refs.~\cite{Borchers:1998ye,Caminiti:2025hjq}.

The paper is organized as follows. Section~\ref{sec:geometry} develops the
recursive geometry and transported states. Section~\ref{sec:free-current}
gives an explicit free-current verification of the KMS condition.
Section~\ref{sec:recursive-bw} proves recursive Bisognano--Wichmann
modularity and  identifies the sibling with the
modular commutant. Section~\ref{sec:generators} develops the recursive
modular generators, their universal near-horizon form, and the associated
stress profiles. Section~\ref{sec:one-generation} establishes the
one-generation-shift identity and distinguishes the observer flow from the
thermal-sector modular flow. Section~\ref{sec:conclusions_outlook} discusses the
conclusions and outlook. Technical details are
collected in the appendices.

\begin{figure}[t]
\centering

\begin{tikzpicture}[
    line width=1.1pt,
    line cap=round,
    line join=round,
    >=stealth
]

\definecolor{BlueRegion}{RGB}{190,220,245}
\definecolor{Done}{RGB}{250,220,175}

\begin{scope}[shift={(0,0)}]

    \draw
        (0,2)
        -- (2,0)
        -- (0,-2)
        -- (-2,0)
        -- cycle;

    \node at (0,-2.45) {$W_0$};

\end{scope}

\draw[->, line width=1.1pt]
    (2.3,0) -- (3,0);

\begin{scope}[shift={(5.35,0)}]

    \fill[BlueRegion]
        (-2,0)
        -- (-1,1)
        -- (0,0)
        -- (-1,-1)
        -- cycle;

    \fill[BlueRegion]
        (0,0)
        -- (1,1)
        -- (2,0)
        -- (1,-1)
        -- cycle;

    \draw
        (0,2)
        -- (2,0)
        -- (0,-2)
        -- (-2,0)
        -- cycle;

    \draw
        (-1,1) -- (1,-1);

    \draw
        (-1,-1) -- (1,1);

    \node at (-0.72,0) {$D_0$};

    \node at (0.95,0) {$W_1$};

\end{scope}

\draw[->, line width=1.1pt]
    (7.80,0) -- (8.5,0);

\begin{scope}[shift={(10.70,0)}]

    \fill[BlueRegion]
        (-2,0)
        -- (-1,1)
        -- (0,0)
        -- (-1,-1)
        -- cycle;

    \fill[Done]
        (0,0)
        -- (0.5,0.5)
        -- (1,0)
        -- (0.5,-0.5)
        -- cycle;

    \fill[BlueRegion]
        (1,0)
        -- (1.5,0.5)
        -- (2,0)
        -- (1.5,-0.5)
        -- cycle;

    \draw
        (0,2)
        -- (2,0)
        -- (0,-2)
        -- (-2,0)
        -- cycle;
        
    \draw
        (-1,1) -- (1,-1);

    \draw
        (-1,-1) -- (1,1);

    \draw
        (0.5,-0.5) -- (1.5,0.5);

    \draw
        (0.5,0.5) -- (1.5,-0.5);

    \node at (-0.62,0) {$D_0$};

    \node at (0.55,0) {$D_1$};

    \node at (1.55,0) {$W_2$};

\end{scope}

\end{tikzpicture}

\caption{Conformal diagram of 2D Minkowski being split into child wedges, In the conformal diagram, all the wedges are self-similar and each split in self-similar to the previous split.This is the geometric intuition behind the results}
\label{fig:nested-rindler}
\end{figure}
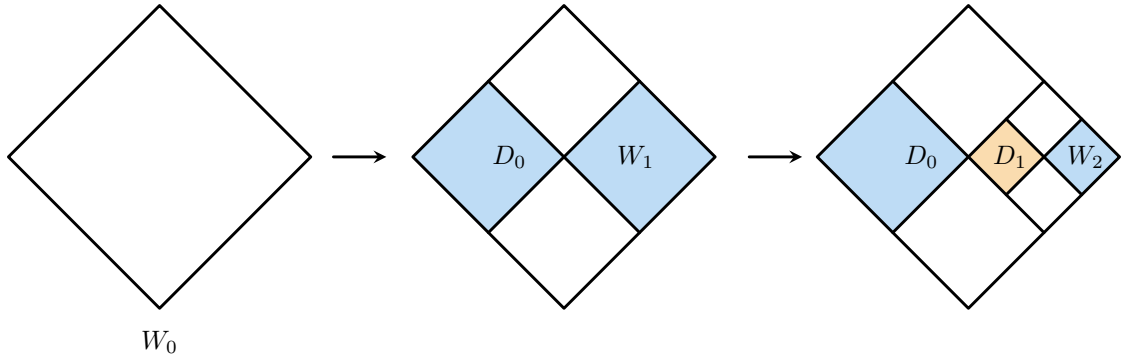

\section{Recursive Geometry and Transported States}
\label{sec:geometry}

We work in two-dimensional Minkowski spacetime $M$, with null coordinates
\begin{equation}
  v=t+x,\qquad u=t-x,\qquad
  \dd s^2=-\dd u\,\dd v.
  \label{eq:null-coordinates}
\end{equation}
In this section, we describe the recursive construction of child regions from
parent regions, illustrated in Fig.~\ref{fig:causal-tree}. At each generation,
a parent region is conformally identified with a reference copy of Minkowski
spacetime. The right and left Rindler wedges of the reference plane then define
two causally complementary children.

\subsection{Parent--Child Construction and State Transport}
\label{subsec:parent-split}

At level $n$, conformal coordinates $(t_n,x_n)$ identify the parent region
$W_n$ with a reference copy of Minkowski spacetime, denoted
$M_{\mathrm{ref}}$. We write
\begin{equation}
  \chi_n:M_{\mathrm{ref}}\longrightarrow W_n,
  \qquad
  v_n=t_n+x_n,\qquad
  u_n=t_n-x_n,
  \label{eq:parent-chart}
\end{equation}
for the orientation- and time-orientation-preserving conformal map. Since the
spacetime argument specifies whether an algebra belongs to the reference plane
or to the transported region, we use the same net symbol $\A$ for both.

In $M_{\mathrm{ref}}$, let
\begin{equation}
  R_n=\{v_n>0,\ u_n<0\},
  \qquad
  L_n=\{v_n<0,\ u_n>0\}
  \label{eq:reference-wedges}
\end{equation}
be the standard right and left Rindler wedges. Their images under $\chi_n$
define the next right child and its sibling:
\begin{equation}
  W_{n+1}:=\chi_n(R_n),
  \qquad
  D_n:=\chi_n(L_n)=W_n\cap W_{n+1}'.
  \label{eq:child-split}
\end{equation}
Here $W_{n+1}'$ denotes the \emph{causal complement} of the spacetime region
$W_{n+1}$, rather than its ordinary set complement. Thus $W_{n+1}$ is the
child on which the next recursive step is performed, while $D_n$ is its
causally complementary sibling within $W_n$.

\begin{figure}[t]
\centering
\label{fig:causal-tree}
\begin{tikzpicture}[
    line width=1.2pt,
    line cap=round,
    line join=round,
    >=stealth
]


\definecolor{WOneBoundary}{RGB}{55,120,220}
\definecolor{WTwoBoundary}{RGB}{145,55,210}
\definecolor{RedBoundary}{RGB}{245,70,90}


\draw[
    WOneBoundary
]
    (-3.5,3.5) -- (3.5,-3.5);

\draw[
    WOneBoundary
]
    (-3.5,-3.5) -- (3.5,3.5);


\draw[
    WTwoBoundary
]
    (2,0) -- (3.6,1.6);

\draw[
    WTwoBoundary
]
    (2,0) -- (3.6,-1.6);


\draw[
    RedBoundary,
    dashed,
    line width=1.3pt
]
    (1,1) -- (2,0);

\draw[
    RedBoundary,
    dashed,
    line width=1.3pt
]
    (1,-1) -- (2,0);


\node at (1,0.1) {$D_1$};

\node at (3.0,0.09) {$W_2$};

\end{tikzpicture}

\caption{The $D_{1}$ sibling causally complements the $W_2$ wedge. This figure shows it in a spacetime diagram.}
\end{figure}
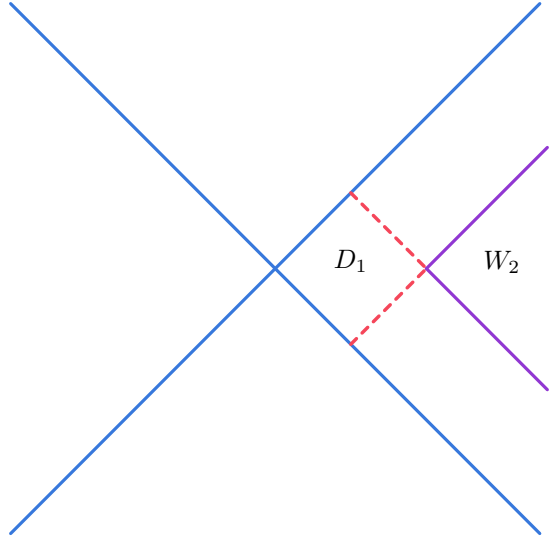
At the root,
\begin{equation}
  W_0=M,\qquad \chi_0=\id,\qquad D_0=W_1'.
  \label{eq:root-data}
\end{equation}
For $n\geq1$, $D_n$ becomes the finite causal diamond between successive
horizons in the explicit realization below. Applying the construction to
every child generates the full binary causal tree; the formulas in this
paper follow the branch in \eqref{eq:intro-chain}.

The standard boost in the reference plane is parametrized by a positive number
$g_{n+1}$,
\begin{equation}
  B_s^{(n+1)}:
  (v_n,u_n)\longmapsto
  \bigl(\e^{g_{n+1}s}v_n,\e^{-g_{n+1}s}u_n\bigr).
  \label{eq:reference-boost}
\end{equation}
The regions in \eqref{eq:reference-wedges}, and hence the child regions, do not
depend on $g_{n+1}$; its role as a clock normalization is explained in
Section~\ref{subsec:clock-normalization}.

\label{subsec:state-transport}
The conformal chart $\chi_n$ also provides a map between observables on the reference Minkowski plane and
observables in the parent region $W_n$. At the level of local algebras, conformal covariance gives the corresponding
local isomorphisms; assuming these maps are chosen coherently under inclusion, they define the quasilocal
$C^*$-isomorphism
\begin{equation}
  \Phi_n:\A(M_{\mathrm{ref}})\longrightarrow \A(W_n).
  \label{eq:Phi-def}
\end{equation}
The technical assumptions underlying this local transport are summarized in Appendix~\ref{app:conformal-net}.

A bounded causal diamond is a region
$O=I^+(p)\cap I^-(q)$ for timelike-separated points $p\prec q$; in null
coordinates it is a product $O=I_v\times I_u$ of bounded open intervals.
For every such $O\subset W_n$ whose closure is compact and lies inside
$W_n$, the chart $\chi_n$ induces an isomorphism from the reference local
algebra of $\chi_n^{-1}(O)$ to $\A(O)$.

Let $\omega_{\mathrm{ref}}$ denote the Minkowski vacuum on the reference plane. The state at level $n$ is defined by transporting
this vacuum to the parent region:
\begin{equation}
  \omega_n:=\omega_{\mathrm{ref}}\circ\Phi_n^{-1}.
  \label{eq:omega-def}
\end{equation}
Thus, for any observable $A\in\A(M_{\mathrm{ref}})$,
\begin{equation}
  \omega_n(\Phi_n(A))=\omega_{\mathrm{ref}}(A).
  \label{eq:transported-observable-identity}
\end{equation}

In this sense, $\omega_n$ is the reference vacuum viewed through the level-$n$ conformal identification.

Primary fields acquire their usual chiral Jacobian
factors, while the stress tensor also contains the anomalous Schwarzian
contribution.

In level-$n$ coordinates, the physical metric on $W_n$ is
\begin{equation}
  \dd s^2=
  \exp\!\left[2\sum_{k=1}^{n}g_kx_k(t_n,x_n)\right]
  \bigl(-\dd t_n^2+\dd x_n^2\bigr),
  \label{eq:recursive-metric}
\end{equation}
where the earlier $x_k$ are recursively expressed in the final coordinates.
For the free scalar, understood modulo its zero mode or on neutral test
functions,
\begin{equation}
  \omega_n\!\left(\phi(t_n,x_n)\phi(t_n',x_n')\right)
  =-\frac{1}{4\pi}
  \log\!\left\{
    \mu^2\left[(x_n-x_n')^2-(t_n-t_n'-\ii0)^2\right]
  \right\}.
  \label{eq:free-scalar-parent}
\end{equation}
The first two states are
\begin{equation}
  \omega_0=\omega_M,\qquad \omega_1=\omega_R,
  \label{eq:first-states}
\end{equation}
where $\omega_R$ is the constructive Fulling vacuum \cite{Fulling:1972md}.
For the free quasi-free current theory this agrees with the usual Fulling
ground-state construction. In a general CFT, equality with every abstract
parent-time ground state is not automatic: the exponential chart is a local
conformal construction rather than a single global vacuum-preserving
conformal symmetry of the compactified theory, and distinct representations
may support distinct ground or thermal states
\cite{Camassa:2011te,Camassa:2011wk,Tanimoto:2018mfn}. The recursive theorem therefore
refers specifically to the transported states in \eqref{eq:omega-def}.

\subsection{First Two Generations and Recursive Flow}
\label{subsec:first-two-levels}

The first Rindler chart is
\begin{equation}
  v=\frac{1}{g_1}\e^{g_1v_1},
  \qquad
  u=-\frac{1}{g_1}\e^{-g_1u_1},
  \qquad
  v_1=t_1+x_1,
  \quad u_1=t_1-x_1.
  \label{eq:first-rindler-chart}
\end{equation}
Equivalently,
\begin{equation}
  t=\frac{\e^{g_1x_1}}{g_1}\sinh(g_1t_1),
  \qquad
  x=\frac{\e^{g_1x_1}}{g_1}\cosh(g_1t_1),
  \label{eq:first-rindler-tx}
\end{equation}
with
\begin{equation}
  \dd s^2=\e^{2g_1x_1}(-\dd t_1^2+\dd x_1^2),
  \qquad
  W_1=\{x>|t|\}=\{v>0,\ u<0\},
  \label{eq:first-rindler-metric}
\end{equation}
and
\begin{equation}
  \zeta_1=\partial_{t_1}=g_1(v\partial_v-u\partial_u).
  \label{eq:zeta1}
\end{equation}
Denote the elementary map by $E_1$. Treating $(t_1,x_1)$ as the reference
plane of the framed parent $W_1$, introduce
\begin{equation}
  v_1=\frac{1}{g_2}\e^{g_2v_2},
  \qquad
  u_1=-\frac{1}{g_2}\e^{-g_2u_2}.
  \label{eq:second-rindler-chart}
\end{equation}
Then
\begin{equation}
  W_2=\left\{v>\frac1{g_1},\ u<-\frac1{g_1}\right\}
      =\left\{x-\frac1{g_1}>|t|\right\},
  \label{eq:W2}
\end{equation}
\begin{equation}
  D_1=\left\{0<v<\frac1{g_1},\ -\frac1{g_1}<u<0\right\},
  \label{eq:D1}
\end{equation}
and the child-time vector in the original coordinates is
\begin{equation}
  \zeta_2=\partial_{t_2}
  =g_2\left[v\log(g_1v)\,\partial_v
          -u\log(-g_1u)\,\partial_u\right].
  \label{eq:zeta2}
\end{equation}
This is an ordinary boost in the parent coordinates $(t_1,x_1)$,
conformally transported to level zero. It is the future-directed geometric
KMS flow of the transported Rindler vacuum on $W_2$; its signed relation to
the standard Tomita modular parameter is fixed in
Appendix~\ref{subsec:tomita-physical-time}. The original Minkowski vacuum on
the same translated wedge instead pairs with the ordinary translated-wedge
Killing boost. This distinction is developed in
Section~\ref{subsec:competing-boost} and is central to the
one-generation-shift theorem.

\label{subsec:general-recursion}
The same construction extends recursively to arbitrary levels. For every $k\geq1$ define the elementary map $E_k$
by
\begin{equation}
  v_{k-1}=\frac{1}{g_k}\e^{g_kv_k},
  \qquad
  u_{k-1}=-\frac{1}{g_k}\e^{-g_ku_k}.
  \label{eq:Ek}
\end{equation}
The framed parent charts are
\begin{equation}
  \chi_n=E_1\circ E_2\circ\cdots\circ E_n,
  \qquad n\geq1.
  \label{eq:chi-composite}
\end{equation}
The inverse maps naturally introduce two chiral families of iterated
logarithms:
\begin{align}
  \Lambda_1^+(v)&=\log(g_1v),
  &
  \Lambda_{k+1}^+(v)
  &=\log\!\left[\frac{g_{k+1}}{g_k}\Lambda_k^+(v)\right],
  \label{eq:Lambda-plus}
  \\
  \Lambda_1^-(u)&=\log(-g_1u),
  &
  \Lambda_{k+1}^-(u)
  &=\log\!\left[\frac{g_{k+1}}{g_k}\Lambda_k^-(u)\right].
  \label{eq:Lambda-minus}
\end{align}
The superscripts $+$ and $-$ label the two null directions. At level $k$
these functions are used only where all preceding logarithms are positive,
and
\begin{equation}
  v_k=\frac{1}{g_k}\Lambda_k^+(v),
  \qquad
  u_k=-\frac{1}{g_k}\Lambda_k^-(u).
  \label{eq:level-coordinates}
\end{equation}

Define
\begin{equation}
  F_k(y)=\frac1{g_k}\e^{g_ky},
  \qquad
  \ell_0=0,
  \qquad
  \ell_n=(F_1\circ F_2\circ\cdots\circ F_n)(0).
  \label{eq:horizon-locations}
\end{equation}
Thus $\ell_n$ is the location of the $n$th recursive horizon. Equivalently,
$\ell_n$ is the unique zero of $\Lambda_n^+$, while $-\ell_n$ is the unique
zero of $\Lambda_n^-$. The nested regions are
\begin{equation}
  W_{n+1}=\{v>\ell_n,\ u<-\ell_n\}
         =\{x-\ell_n>|t|\},
  \label{eq:Wn-region}
\end{equation}
\begin{equation}
  D_n=\{\ell_{n-1}<v<\ell_n,\ -\ell_n<u<-\ell_{n-1}\},
  \qquad n\geq1.
  \label{eq:Dn-region}
\end{equation}

The chain rule gives the level-$(n+1)$ child flow
\begin{equation}
  \zeta_{n+1}=g_{n+1}\left[
    v\prod_{j=1}^{n}\Lambda_j^+(v)\,\partial_v
    -u\prod_{j=1}^{n}\Lambda_j^-(u)\,\partial_u
  \right].
  \label{eq:recursive-flow}
\end{equation}
Its apparent nonlinearity is purely a coordinate effect. Set
\begin{equation}
  y_n^+(v)=\Lambda_n^+(v)=g_nv_n,
  \qquad
  y_n^-(u)=\Lambda_n^-(u)=-g_nu_n.
  \label{eq:linearizing-coordinates}
\end{equation}
Repeated differentiation gives
\begin{equation}
  \frac{\dd\Lambda_n^+}{\dd v}
  =\frac{1}{v\prod_{j=1}^{n-1}\Lambda_j^+(v)},
  \qquad
  \frac{\dd\Lambda_n^-}{\dd u}
  =\frac{1}{u\prod_{j=1}^{n-1}\Lambda_j^-(u)}.
  \label{eq:Lambda-derivatives}
\end{equation}
Substituting into \eqref{eq:recursive-flow}, all preceding logarithms cancel:
\begin{equation}
  \zeta_{n+1}y_n^+=g_{n+1}y_n^+,
  \qquad
  \zeta_{n+1}y_n^-=-g_{n+1}y_n^-.
  \label{eq:linearized-flow}
\end{equation}
Thus
\begin{equation}
  y_n^+\longmapsto\e^{g_{n+1}s}y_n^+,
  \qquad
  y_n^-\longmapsto\e^{-g_{n+1}s}y_n^-,
  \label{eq:linearized-boost}
\end{equation}
which is an ordinary Lorentz boost. This places the recursive family within the broader class of conformally
related geometric modular flows in two dimensions \cite{Caminiti:2025hjq}; the distinctive feature here is its closure under the
explicit exponential recursion.

\subsection{Clock normalization}
\label{subsec:clock-normalization}

For a fixed inclusion $W_n\supset W_{n+1}$, $g_{n+1}$ only normalizes the
child clock. It leaves the state and child region fixed while rescaling the
generator and KMS inverse temperature oppositely. If the same geometric flow
is reparametrized so that its generator is multiplied by $c>0$, the numerical
temperature is multiplied by $c$ and the inverse temperature divided by $c$.
The invariant content is the state--flow pairing and its KMS relation, not a
numerical temperature without a clock convention.

In the full recursion, the same parameter becomes part of the next parent
frame and hence helps locate $W_{n+2}$. For $n\geq1$, the transported state
$\omega_n$ is independent of the terminal scale $g_n$. Holding
$g_1,\ldots,g_{n-1}$ fixed and replacing $g_n$ by $g_n'$ gives
\begin{equation}
  v_n'=\frac{g_n}{g_n'}v_n
       +\frac{1}{g_n'}\log\frac{g_n'}{g_n},
  \qquad
  u_n'=\frac{g_n}{g_n'}u_n
       -\frac{1}{g_n'}\log\frac{g_n'}{g_n}.
  \label{eq:terminal-scale-change}
\end{equation}
This is a global reference dilation followed by a spatial translation, under which the conformal vacuum is invariant.
Thus changing the terminal scale $g_n$ moves the next recursive cut without changing the transported state $\omega_n$.

\section{Explicit KMS Verification on the Child Wedge}
\label{sec:free-current}

We now verify the first nontrivial child-wedge KMS property directly for a
two-dimensional free massless scalar. Write
\begin{equation}
  \phi(t_1,x_1)=\phi_+(v_1)+\phi_-(u_1).
  \label{eq:chiral-decomposition}
\end{equation}
The undifferentiated scalar has an infrared zero-mode ambiguity in two
dimensions, so we use the chiral currents
\begin{equation}
  j_{1,+}(v_1)=\partial_{v_1}\phi_+(v_1),
  \qquad
  j_{1,-}(u_1)=\partial_{u_1}\phi_-(u_1).
  \label{eq:chiral-currents}
\end{equation}
Equivalently, one may use the current Weyl algebra or smear the scalar only
with neutral test functions \cite{Derezinski:2004xz}. By
\eqref{eq:omega-def}, $\omega_R=\omega_1$ is the reference Minkowski vacuum
in the $(t_1,x_1)$ frame. Its nonzero chiral two-point distributions are
\begin{align}
  \omega_R\!\left(j_{1,+}(v_1)j_{1,+}(v_1')\right)
  &=-\frac{1}{4\pi}\frac{1}{(v_1-v_1'-\ii0)^2},
  \label{eq:parent-plus-current}
  \\
  \omega_R\!\left(j_{1,-}(u_1)j_{1,-}(u_1')\right)
  &=-\frac{1}{4\pi}\frac{1}{(u_1-u_1'-\ii0)^2},
  \label{eq:parent-minus-current}
\end{align}
with vanishing mixed correlator. The $\ii0$ prescriptions specify
distributional boundary values and encode the operator ordering.

On $W_2$ use
\begin{equation}
  v_1=f(v_2)=\frac1{g_2}\e^{g_2v_2},
  \qquad
  u_1=h(u_2)=-\frac1{g_2}\e^{-g_2u_2},
  \qquad
  v_2=t_2+x_2,
  \quad u_2=t_2-x_2.
  \label{eq:second-map-free}
\end{equation}
Because the currents have conformal weight one,
\begin{equation}
  j_{2,+}(v_2)=f'(v_2)j_{1,+}(f(v_2)),
  \qquad
  j_{2,-}(u_2)=h'(u_2)j_{1,-}(h(u_2)).
  \label{eq:current-transform}
\end{equation}
Using
\begin{align}
  f(v_2)-f(v_2')
  &=\frac2{g_2}\e^{\frac{g_2}{2}(v_2+v_2')}
    \sinh\!\left[\frac{g_2}{2}(v_2-v_2')\right],
  \label{eq:f-difference}
  \\
  h(u_2)-h(u_2')
  &=\frac2{g_2}\e^{-\frac{g_2}{2}(u_2+u_2')}
    \sinh\!\left[\frac{g_2}{2}(u_2-u_2')\right],
  \label{eq:h-difference}
\end{align}
we obtain, as distributions,
\begin{align}
  \omega_R\!\left(j_{2,+}(v_2)j_{2,+}(v_2')\right)
  &=-\frac{g_2^2}{16\pi}
  \frac{1}{\sinh^2\!\left[\frac{g_2}{2}(v_2-v_2'-\ii0)\right]},
  \label{eq:child-plus-current}
  \\
  \omega_R\!\left(j_{2,-}(u_2)j_{2,-}(u_2')\right)
  &=-\frac{g_2^2}{16\pi}
  \frac{1}{\sinh^2\!\left[\frac{g_2}{2}(u_2-u_2'-\ii0)\right]}.
  \label{eq:child-minus-current}
\end{align}
These formulas hold for arbitrary points in $W_2$, not merely along one
fixed-$x_2$ orbit. With
\begin{equation}
  \Delta t_2=t_2-t_2',\qquad \Delta x_2=x_2-x_2',
  \label{eq:deltas}
\end{equation}

and $\partial_{t_2}\phi=j_{2,+}+j_{2,-}$,
\begin{align}
  &\omega_R\!\left(
    \partial_{t_2}\phi(t_2,x_2)
    \partial_{t_2}\phi(t_2',x_2')
  \right)
  \nonumber\\
  &\qquad=-\frac{g_2^2}{16\pi}\left\{
  \frac{1}{\sinh^2\!\left[\frac{g_2}{2}(\Delta t_2+\Delta x_2-\ii0)\right]}
  +
  \frac{1}{\sinh^2\!\left[\frac{g_2}{2}(\Delta t_2-\Delta x_2-\ii0)\right]}
  \right\}.
  \label{eq:full-child-correlation}
\end{align}
At $x_2=x_2'$ this reduces to
\begin{equation}
  \omega_R\!\left(
    \partial_{t_2}\phi(t_2,x_2)
    \partial_{t_2}\phi(t_2',x_2)
  \right)
  =-\frac{g_2^2}{8\pi}
  \frac{1}{\sinh^2\!\left[\frac{g_2}{2}(\Delta t_2-\ii0)\right]}.
  \label{eq:orbit-correlation}
\end{equation}

\paragraph{Analytic strip and operator ordering}
\label{subsec:analytic-strip}
Set
\begin{equation}
  \beta_2=\frac{2\pi}{g_2}.
  \label{eq:beta2}
\end{equation}

For fixed $x_2$ and $x_2'$, define the complex-time function
\begin{equation}
  C_{x_2,x_2'}(z)=-\frac{g_2^2}{16\pi}\left[
  \frac{1}{\sinh^2\!\left(\frac{g_2}{2}(z+\Delta x_2)\right)}
  +\frac{1}{\sinh^2\!\left(\frac{g_2}{2}(z-\Delta x_2)\right)}
  \right],
  \label{eq:current-complex-time-function}
\end{equation}
where $\Delta x_2=x_2-x_2'$ and the boundary-value prescription has been temporarily omitted. The possible singularities
occur when
\[
  \frac{g_2}{2}(z\pm\Delta x_2)=\ii\pi n,\qquad n\in\mathbb{Z},
\]
or equivalently
\begin{equation}
  z=\mp\Delta x_2+\ii n\beta_2.
  \label{eq:current-correlation-poles}
\end{equation}
Hence there are no singularities in the open strip
\begin{equation}
  -\beta_2<\operatorname{Im}z<0.
  \label{eq:lower-strip}
\end{equation}
The boundary lines themselves contain the singularities of the distributional correlation function and are
understood through the appropriate boundary-value prescriptions.

Let
\begin{equation}
  O(t_2,x_2)=\partial_{t_2}\phi(t_2,x_2).
  \label{eq:current-time-observable}
\end{equation}
Approaching the upper boundary of the strip from below, we set $z=t-\ii\epsilon$ with $\epsilon>0$. Then
\begin{equation}
  \lim_{\epsilon\downarrow0}C_{x_2,x_2'}(t-\ii\epsilon)
  =\omega_R\!\left(O(t,x_2)O(0,x_2')\right),
  \label{eq:upper-boundary}
\end{equation}
where time-translation invariance has been used to set the second time coordinate to zero. Thus the upper
boundary reproduces the original Wightman ordering.

For the lower boundary, approach the line $\operatorname{Im}z=-\beta_2$ from within the strip:
\[
  z=t-\ii\beta_2+\ii\epsilon,\qquad \epsilon>0.
\]
Since
\[
  \frac{g_2\beta_2}{2}=\pi,
\]
each hyperbolic sine in \eqref{eq:current-complex-time-function} is shifted by $-\ii\pi$. Using
\[
  \sinh(w-\ii\pi)=-\sinh w,
\]
and noting that the square removes the minus sign, one obtains
\begin{equation}
  \lim_{\epsilon\downarrow0}C_{x_2,x_2'}(t-\ii\beta_2+\ii\epsilon)
  =\omega_R\!\left(O(0,x_2')O(t,x_2)\right).
  \label{eq:lower-boundary}
\end{equation}

The crucial point is that the shift by $-\ii\beta_2$ returns the same meromorphic function but changes the boundary
prescription from $t-\ii0$ to $t+\ii0$. This corresponds to the reversal of the Wightman operator ordering.

Equations \eqref{eq:upper-boundary} and \eqref{eq:lower-boundary}, together with analyticity in the open strip, are precisely the lower-strip KMS
boundary conditions \cite{Haag:1967sg,Summers:2003tf}. Thus the explicit two-point function satisfies the KMS condition with respect to
$t_2$-translations, with
\begin{equation}
  \beta_2=\frac{2\pi}{g_2},\qquad T_2=\frac{g_2}{2\pi}.
  \label{eq:free-KMS-temperature}
\end{equation}
The displayed currents are unbounded operator-valued distributions, so the fully algebraic statement is obtained
after smearing them with test functions and passing to the corresponding current Weyl algebra. The standard
quasi-free KMS criterion then promotes the above two-point function result to the bounded local Weyl algebra
\cite{Bratteli:1979tw}.

After smearing with compactly supported real test functions in $W_2$,
\eqref{eq:child-plus-current} and \eqref{eq:child-minus-current} determine
the covariance of a quasi-free state on the local CCR/Weyl algebra generated
by the smeared chiral currents. Let $W(f)$ denote the corresponding Weyl
operators, and let $\alpha_s^{(2)}$ be the automorphism induced by
$t_2$-translation. The smeared covariance is stationary under
$\alpha_s^{(2)}$ and has the analytic continuation and boundary ordering in
\eqref{eq:upper-boundary}--\eqref{eq:lower-boundary}.

By the standard quasi-free KMS criterion for CCR algebras
\cite{Bratteli:1996xq}, this covariance condition is equivalent to the
KMS condition for the associated bounded Weyl algebra. For compactly
supported $f$ and $g$, the real-time function
\begin{equation}
  F_{f,g}(t):=
  \omega_R\!\left(\alpha_t^{(2)}(W(f))W(g)\right),
  \qquad t\in\R,
  \label{eq:Weyl-real-time-function}
\end{equation}
has a bounded analytic continuation to
$-\beta_2<\operatorname{Im}z<0$ and satisfies
\begin{equation}
  F_{f,g}(t-\ii\beta_2)
  =\omega_R\!\left(W(g)\alpha_t^{(2)}(W(f))\right).
  \label{eq:Weyl-KMS-boundary}
\end{equation}
All $\ii0$ prescriptions denote distributional boundary values
and are not finite regulators transported literally through the exponential
map.

\par\medskip\noindent\hspace{\parindent}\textbf{Conformal time is not fixed-orbit proper time}\par
\label{subsec:proper-time}
After two transformations,
\begin{equation}
  \dd s^2=
  \e^{2g_1x_1(t_2,x_2)+2g_2x_2}
  (-\dd t_2^2+\dd x_2^2),
  \qquad
  x_1(t_2,x_2)=\frac{\e^{g_2x_2}}{g_2}\cosh(g_2t_2).
  \label{eq:second-metric}
\end{equation}
Along a fixed-$x_2$ curve,
\begin{equation}
  \dd\tau=\e^{g_1x_1(t_2,x_2)+g_2x_2}\,\dd t_2,
  \label{eq:proper-time-lapse}
\end{equation}
and the factor is time dependent. Hence $\partial_{t_2}\phi$ is not the
proper-time derivative $\dd\phi/\dd\tau$, and a Fourier frequency conjugate
to $t_2$ is not automatically a detector energy gap. Kolekar and
Padmanabhan explicitly noted this distinction and considered a nonconstant-
$x_2$ trajectory for which $t_2$ may be chosen as proper time
\cite{Kolekar:2013hra}. More recent analyses treat
trajectory-dependent and asymptotic detector responses
\cite{Sahota:2024thk,Dubey:2025hab}.

The exact conclusion is therefore state-level KMS thermality with respect to
the child conformal/modular clock. A proper-time detector temperature is a
separate operational statement. Since the lapse in
\eqref{eq:proper-time-lapse} varies along a fixed-$x_2$ trajectory, no
stationary Tolman redshift converts $g_2/(2\pi)$ into a single detector
temperature on that trajectory.
\section{Recursive Bisognano--Wichmann modularity}
\label{sec:recursive-bw}

Section~\ref{sec:geometry} constructed the recursive family of parent and child regions, together with the conformal transport of
the reference vacuum to each parent. We now use that construction to establish the corresponding modular
statement. The basic mechanism is simple: the Bisognano--Wichmann property is applied to the standard right
wedge in the reference Minkowski plane, and the resulting state--flow pairing is transported to the recursive child.
The recursive KMS property therefore follows from ordinary wedge modularity together with the conformal
transport.

\subsection{Quasilocal and GNS Transport}
\label{subsec:quasilocal-transport}

The recursive construction of Section~\ref{sec:geometry} was formulated
first at the level of quasilocal $C^*$-algebras. We now pass to the
von Neumann algebras and GNS representations needed for modular theory.

For an open region $R$, the quasilocal algebra is
\begin{equation}
  \A(R)=\overline{\bigcup_{O\subset R}\A(O)},
  \label{eq:quasilocal-algebra}
\end{equation}
where the union runs over bounded causal diamonds $O$ whose compact
closures lie inside $R$, and the overline denotes closure in the operator
norm. For a state $\omega$ with GNS triple
$(\pi_\omega,\Hh_\omega,\Omega_\omega)$, we define
\begin{equation}
  M_\omega(R):=
  \overline{\pi_\omega(\A(R))}^{\,\mathrm{wot}}.
  \label{eq:von-neumann-closure}
\end{equation}

By the local conformal covariance and coherence assumptions stated in
Appendix~\ref{app:conformal-net}, the local transport extends to the
quasilocal isomorphism
\begin{equation}
  \Phi_n:\A(M_{\mathrm{ref}})\longrightarrow\A(W_n),
  \qquad
  \Phi_n\!\left(\A(R_n)\right)=\A(W_{n+1}).
  \label{eq:Phi-child}
\end{equation}
The state introduced in Section~\ref{sec:geometry} is
\begin{equation}
  \omega_n=\omega_{\mathrm{ref}}\circ\Phi_n^{-1}.
  \label{eq:v2-transported-state-restatement}
\end{equation}

The corresponding GNS representations are related by the canonical
unitary
\begin{equation}
  V_n:\Hh_{\mathrm{ref}}\longrightarrow\Hh_{\omega_n},
\end{equation}
defined on the dense set of GNS vectors by
\begin{equation}
  V_n\pi_{\mathrm{ref}}(A)\Omega_{\mathrm{ref}}
  =
  \pi_{\omega_n}\!\left(\Phi_n(A)\right)\Omega_{\omega_n}.
  \label{eq:v2-reference-gns-unitary}
\end{equation}
The well-definedness and unitary extension of this map are given in
Appendix~\ref{subsec:GNS-transport}.

At the von Neumann level, the same unitary implements the transport
\begin{equation}
  \Theta_n:=\operatorname{Ad}V_n:
  M_{\omega_{\mathrm{ref}}}(R_n)
  \longrightarrow
  M_{\omega_n}(W_{n+1}),
  \qquad
  \Theta_n(X)=V_nXV_n^*.
  \label{eq:Theta-n}
\end{equation}
Consequently, whenever the reference wedge is in standard form, its
Tomita--Takesaki objects are transported to the child algebra as
\begin{equation}
  S_n=V_nS_R^{\mathrm{ref}}V_n^*,
  \qquad
  \Delta_{n\to n+1}
  =V_n\Delta_R^{\mathrm{ref}}V_n^*,
  \qquad
  J_n=V_nJ_R^{\mathrm{ref}}V_n^*.
  \label{eq:v2-modular-objects-transport}
\end{equation}
The general GNS and modular-transport statements are collected in
Appendix~\ref{subsec:GNS-transport}.

Finally, the reference boost induces an automorphism group on the reference
wedge algebra, denoted by $B_s^{(n+1)}$. Its transport to the child algebra is
\begin{equation}
  \alpha_s^{(n+1)}
  :=
  \Phi_n\circ B_s^{(n+1)}\circ\Phi_n^{-1}.
  \label{eq:transported-child-group}
\end{equation}
This formulation transports the dynamics at the level of automorphism
groups rather than applying $\Phi_n$ directly to an unbounded generator.
The Bisognano--Wichmann property can therefore be transported from the
reference wedge to the child algebra, with the KMS convention and physical
clock normalization fixed as in Appendix~\ref{subsec:tomita-physical-time}.

\subsection{Recursive Bisognano--Wichmann theorem}
\label{subsec:main-theorem}

\paragraph{Standing assumptions.}
We work with a bosonic, local, positive-energy, diffeomorphism-covariant two-dimensional conformal net. We assume that the conformal maps used in the recursive construction admit
coherent local implementation on the net, so that they induce the quasilocal isomorphisms used above. We
also assume the vacuum Bisognano--Wichmann property for the standard wedges in the reference Minkowski
plane. Wedge duality is not needed for the KMS statement below; it will be used later to identify the causally
complementary sibling with the modular commutant in Section~\ref{sec:recursive-bw}.

These assumptions, rather than the coordinate recursion by itself, define the class of theories to which the
theorem applies. Examples include vacuum-sector tensor products of diffeomorphism-covariant chiral conformal
nets and suitable local two-dimensional extensions satisfying the required covariance and duality properties
\cite{Bisognano:1975ih,Bisognano:1976za,Brunetti:1992zf}. The detailed assumptions concerning local covariance, transport, and the Bisognano--Wichmann
property are collected in Appendix~\ref{app:conformal-net}.

\begin{theorem}
\label{thm:recursive-bw}
Let the net satisfy the standing assumptions, and let $(W_n,\chi_n)$ be a framed parent with transported
state $\omega_n$ defined by \eqref{eq:omega-def}. Then

\begin{equation}
  \left.\omega_n\right|_{\A(W_{n+1})}
  \text{ is KMS for }\alpha_s^{(n+1)}
  \quad\text{at}\quad
  \beta_{n+1}=\frac{2\pi}{g_{n+1}}.
  \label{eq:recursive-bw-statement}
\end{equation}
\end{theorem}

Equivalently, the restriction of the transported parent state to the next right child is thermal with respect to the
transported child boost.

\begin{proof}
Consider the reference Minkowski plane $M_{\mathrm{ref}}$ associated with the framed parent $(W_n,\chi_n)$. By the
Bisognano--Wichmann property, the reference vacuum restricted to the standard right-wedge algebra $\A(R_n)$ is a
KMS state for the reference boost $B_s^{(n+1)}$.

The child region is the image of this reference wedge,
\[
  W_{n+1}=\chi_n(R_n),
\]
and the corresponding algebra is related by
\[
  \A(W_{n+1})=\Phi_n\!\left(\A(R_n)\right).
\]

At the same time, the state on the parent and the child flow are obtained by transporting the reference state
and reference boost through the same isomorphism:
\begin{equation}
  \omega_n=\omega_{\mathrm{ref}}\circ\Phi_n^{-1},
  \qquad
  \alpha_s^{(n+1)}=\Phi_n\circ B_s^{(n+1)}\circ\Phi_n^{-1}.
  \label{eq:v2-state-flow-common-transport}
\end{equation}

The KMS condition is invariant under a $C^*$-algebra isomorphism. Hence the KMS property of the reference
vacuum on $\A(R_n)$ transfers directly to
\[
  \omega_n|_{\A(W_{n+1})}
\]
with respect to $\alpha_s^{(n+1)}$.

The reference boost parameter is the dimensionless rapidity $g_{n+1}t_{n+1}$. Since the Bisognano--Wichmann inverse
temperature is $2\pi$ when expressed in this rapidity parameter, the inverse temperature with respect to the child
coordinate time $t_{n+1}$ is
\[
  \beta_{n+1}=\frac{2\pi}{g_{n+1}}.
\]
This proves the theorem.
\end{proof}

\subsection{Transported standard form, purity, and representation}
\label{subsec:transported-standard-form}

If the reference vacuum representation is irreducible,
$\omega_{\mathrm{ref}}$ is pure, and this purity is preserved under the
$C^*$-isomorphisms used in the recursive construction. Hence, for $n\geq1$,
the transported state $\omega_n$ is pure on the parent quasilocal algebra
$\A(W_n)$. This purity refers to the transported parent state; its
restriction to the child algebra is the state whose KMS property is
established below.

For every causal diamond $O$ whose compact closure lies inside $W_n$, a
circle-diffeomorphism extension of the chart implements
$\omega_n|_{\A(O)}$ as a vector state in the original vacuum representation,
so the transported state is locally normal. We do not need to classify the
global representation of the full parent wedge for the modular statements
below.

All modular objects below therefore refer to the transported standard form
\begin{equation}
  \bigl(M_{\omega_n}(W_{n+1}),\Omega_{\omega_n}\bigr)
  \quad\text{on}\quad \Hh_{\omega_n}.
  \label{eq:transported-standard-pair}
\end{equation}

The modular flow and state considered here are understood within this
specified transported GNS representation. In particular, the pairing of
a state with its modular dynamics is representation-dependent: across
inequivalent representations or sectors, the same geometric dynamics may
admit distinct KMS states
\cite{Camassa:2011te,Camassa:2011wk}.
Within the recursive construction, the linearizing chart selects the
transported state together with its GNS representation and normalized
modular flow.

\subsection{Modular conjugation and the sibling algebra}
\label{subsec:conjugation-sibling}

A prime on an operator algebra denotes its commutant:

\begin{equation}
  N':=\{X\in B(\Hh):[X,A]=0\ \text{for every }A\in N\}.
  \label{eq:commutant-definition}
\end{equation}

By the bicommutant theorem, the weak closure in
\eqref{eq:von-neumann-closure} also equals $\pi_\omega(\A(R))''$.

Let
\begin{equation}
  M_R^{\mathrm{ref}}:=M_{\omega_{\mathrm{ref}}}(R_n),
  \qquad
  M_L^{\mathrm{ref}}:=M_{\omega_{\mathrm{ref}}}(L_n).
  \label{eq:reference-MRL}
\end{equation}

For the conformal observable nets considered here, wedge duality gives
\begin{equation}
  (M_R^{\mathrm{ref}})'=M_L^{\mathrm{ref}}.
  \label{eq:reference-wedge-duality}
\end{equation}

This is the duality input associated with the conformal Bisognano--Wichmann/PCT structure \cite{Brunetti:1992zf}.

Transporting the two wedges to the $n$-th state gives

\begin{equation}
  M_R^n:=M_{\omega_n}(W_{n+1}),
  \qquad
  M_L^n:=M_{\omega_n}(D_n),
  \label{eq:MRML}
\end{equation}

and the transported wedge duality relation is

\begin{equation}
  M_L^n=(M_R^n)'.
  \label{eq:transported-wedge-duality}
\end{equation}

If $J_n$ denotes the modular conjugation of the standard pair $(M_R^n,\Omega_{\omega_n})$, Tomita--Takesaki theory then gives

\begin{equation}
  J_nM_R^nJ_n=(M_R^n)'=M_L^n.
  \label{eq:modular-purifier}
\end{equation}
Thus modular conjugation maps the right-child algebra onto its commutant, while wedge duality identifies that
commutant geometrically with the algebra of the complementary sibling $D_n$. By contrast, locality alone gives
only the inclusion

\begin{equation}
  M_{\omega_n}(D_n)\subseteq M_{\omega_n}(W_{n+1})'.
  \label{eq:locality-inclusion}
\end{equation}

Related geometric interpretations of commutants beyond thermal horizons were developed in Ref.~\cite{Schroer:1998pj}.

The geometric action is explicit at the first nested step. In the reference $(t_1,x_1)$ plane, modular conjugation
contains the wedge reflection

\begin{equation}
  (t_1,x_1)\longmapsto(-t_1,-x_1),
  \qquad
  (v_1,u_1)\longmapsto(-v_1,-u_1),
  \label{eq:wedge-reflection}
\end{equation}

 together with the theory-dependent antiunitary action on internal field labels. Using the first Rindler chart
\eqref{eq:first-rindler-chart}, this reflection becomes, in the original coordinates,

\begin{equation}
  v\longmapsto\frac1{g_1^2v},
  \qquad
  u\longmapsto\frac1{g_1^2u}.
  \label{eq:geometric-J-action}
\end{equation}

For $v>1/g_1$, the image satisfies $0<v'<1/g_1$, while $u<-1/g_1$ maps to $-1/g_1<u'<0$. Hence the
transformation exchanges $W_2$ with $D_1$. Since it is an involution, applying it twice returns the original point.

Equation~\eqref{eq:modular-purifier} gives the exact continuum modular-commutant relation underlying the purification interpretation.
For type-III local algebras, however, this relation does not imply an ordinary Hilbert-space tensor-product
factorization. If the relevant inclusion satisfies the split property, one may introduce an intermediate type-I
factor and a corresponding tensor-product representation, for example by separating bounded causal-diamond
approximants by a nonzero collar \cite{Doplicher:1984zz,Witten:2018zxz}. In such a split-regulated representation, the parent vector can be
written in Schmidt form,

\begin{equation}
  |0\rangle_{W_n}
  \sim\sum_\alpha\sqrt{\lambda_\alpha}
  |\alpha\rangle_{W_{n+1}}|\alpha\rangle_{D_n},
  \qquad
  \lambda_\alpha\geq0,
  \quad
  \sum_\alpha\lambda_\alpha=1.
  \label{eq:split-Schmidt}
\end{equation}

The symbol $\sim$ emphasizes both the regulator and the antiunitary
identification of the partner basis.

The Schmidt weights may be parametrized as
\begin{equation}
  \lambda_\alpha=Z^{-1}\e^{-\beta_{n+1}E_\alpha},
  \label{eq:v2-regulated-Schmidt-weights}
\end{equation}

thereby defining regulated entanglement energies. This parametrization does not by itself identify the $E_\alpha$ with
the spectrum of a physical child-time Hamiltonian; such an identification requires additional dynamical input.
Thus the thermofield interpretation is a split-regulated realization of the exact modular-commutant pairing \cite{Israel:1976ur},
rather than a tensor-product factorization of the underlying type-III algebras.

The two child regions alone do not determine the parent state:
the corresponding algebra embeddings, representation, and state vector
are also part of the modular construction.

\section{Modular generators, horizon behavior, and stress profiles}
\label{sec:generators}

\subsection{Recursive Modular Generators}
\label{subsec:one-sided-generator}

There are two related quantities that are often referred to as the modular Hamiltonian. The exact modular
generator on the full GNS Hilbert space is the modular Liouvillean

\begin{equation}
  L_{n\to n+1}:=-\log\Delta_{n\to n+1}.
  \label{eq:Liouvillean}
\end{equation}

It acts on the full GNS Hilbert space and is not, in general, an operator belonging to the right-child algebra $M_R^n$.

We denote by $K^R_{n\to n+1}$ the corresponding one-sided stress-tensor charge. Its commutator gives the modular
flow on suitable right-child observables. In the continuum theory, where local algebras are type III, $K^R_{n\to n+1}$
should be understood as a quadratic form or as a derivation rather than as an operator affiliated with $M_R^n$. If a
split regularization is introduced, it can instead be represented as an operator on the right type-I factor. In that
setting one may write schematically

\begin{equation}
  L_{n\to n+1}=K^R_{n\to n+1}-K^L_{n\to n+1},
  \qquad
  K^L_{n\to n+1}\ \text{affiliated with }(M_R^n)'.
  \label{eq:L-KR-KL}
\end{equation}

The commutant term does not affect the adjoint action on $M_R^n$. Thus, the stress-tensor expressions below
represent the one-sided modular generator, while the exact continuum implementation is the Liouvillean $L_{n\to n+1}$.

In the reference level-$n$ coordinates, the right-wedge generator at $t_n=0$ is

\begin{equation}
  K^R_{\mathrm{ref}}
  =2\pi\int_0^\infty\dd x_n\,x_n
  T^{\mathrm{ref}}_{t_nt_n}(0,x_n),
  \label{eq:reference-KR-tx}
\end{equation}

Using
\[
  v_n=x_n,\qquad u_n=-x_n\qquad(t_n=0),
\]
this can equivalently be written in null coordinates as

\begin{equation}
  K^R_{\mathrm{ref}}
  =2\pi\int_0^\infty\dd v_n\,v_nT^{\mathrm{ref}}_{v_nv_n}
   +2\pi\int_{-\infty}^0\dd u_n\,(-u_n)T^{\mathrm{ref}}_{u_nu_n}.
  \label{eq:reference-KR-null}
\end{equation}

The formulas below use the state-subtracted operator part of the stress
tensor in the transported representation.

In the following, we keep only the operator (tensorial) part of the transformed stress tensor. Under a conformal
transformation, the stress tensor also acquires an anomalous Schwarzian term. This term is a $c$-number, so it
does not affect the commutator with right-child observables. We absorb its contribution to the charge into $C_n$.
Since this contribution is a $c$-number, it does not affect the adjoint action on right-child observables.

For $n\geq1$, the recursive coordinate transformation gives

\begin{equation}
  \frac{\dd v}{\dd v_n}
  =g_nv\prod_{j=1}^{n-1}\Lambda_j^+(v),
  \qquad
  \frac{\dd u}{\dd u_n}
  =-g_nu\prod_{j=1}^{n-1}\Lambda_j^-(u).
  \label{eq:inverse-jacobians}
\end{equation}

For the tensorial part of the stress tensor, the null-null component transforms with two Jacobian factors,
\begin{equation}
  T^{\mathrm{ref}}_{v_nv_n}
  =\left(\frac{\dd v}{\dd v_n}\right)^2T_{vv},
  \label{eq:stress-tensor-jacobian}
\end{equation}
and similarly for the $u$ component. Since the charge also contains the measure $\dd v_n$, one of these Jacobian factors
is cancelled when we change the integration variable from $v_n$ to $v$. Thus
\begin{equation}
  v_nT^{\mathrm{ref}}_{v_nv_n}\,\dd v_n
  =v_n\frac{\dd v}{\dd v_n}T_{vv}(v)\,\dd v.
  \label{eq:stress-measure-jacobian}
\end{equation}
Using
\begin{equation}
  v_n=\frac1{g_n}\Lambda_n^+(v)
  \label{eq:stress-null-log}
\end{equation}
together with Eq.~\eqref{eq:inverse-jacobians}, we obtain

\begin{equation}
  v_nT^{\mathrm{ref}}_{v_nv_n}\,\dd v_n
  =v\prod_{j=1}^{n}\Lambda_j^+(v)T_{vv}(v)\,\dd v,
  \label{eq:stress-transform-plus}
\end{equation}

with the analogous relation for the $u$ component.

Therefore,

\begin{align}
  K^R_{n\to n+1}
  &=2\pi\int_{\ell_n}^{\infty}\dd v\,
    v\prod_{j=1}^{n}\Lambda_j^+(v)T_{vv}(v)
  \nonumber\\
  &\quad+2\pi\int_{-\infty}^{-\ell_n}\dd u\,
    (-u)\prod_{j=1}^{n}\Lambda_j^-(u)T_{uu}(u)
  +C_n.
  \label{eq:recursive-KR}
\end{align}

Here $C_n$ collects the $c$-number contribution from the anomalous Schwarzian term, together with any additive
convention for the generator. At the first nested step,

\begin{align}
  K^R_{1\to2}
  &=2\pi\int_{1/g_1}^{\infty}\dd v\,
    v\log(g_1v)T_{vv}(v)
  \nonumber\\
  &\quad+2\pi\int_{-\infty}^{-1/g_1}\dd u\,
    (-u)\log(-g_1u)T_{uu}(u)+C_1.
  \label{eq:first-recursive-KR}
\end{align}

If $\Sigma$ is the $t_{n+1}=0$ Cauchy slice of the child wedge, with future-directed $\dd\Sigma^\nu$, the same result can be written
covariantly as

\begin{equation}
  K^R_{n\to n+1}
  =\frac{2\pi}{g_{n+1}}
  \int_\Sigma T_{\mu\nu}\zeta_{n+1}^\mu\,\dd\Sigma^\nu+C_n.
  \label{eq:covariant-KR}
\end{equation}

Thus $(g_{n+1}/2\pi)K^R_{n\to n+1}$ is the one-sided stress-tensor generator of physical child-time evolution at the quadratic-form or split-regulated level. The exact continuum implementation is instead given by $(g_{n+1}/2\pi)L_{n\to n+1}$. In a
split representation, their commutators with right-child observables agree because their difference is affiliated
with the commutant.

The local stress-tensor form places the recursive family within the broad class of geometric modular flows in
two-dimensional CFT. Constraints on geometric modular flow were studied
in Ref.~\cite{Sorce2024}, where its conformal character and
its future-directedness in weakly analytic states were
established under the assumptions considered there. General conformal-flow constructions, state reconstruction, and energy diagnostics were
developed in Ref.~\cite{Caminiti:2025hjq}. Here the distinctive result is the closed iterated-log family generated by a self-similar
sequence of child embeddings.

\subsection{Universal horizon linearization and the competing boost}
\label{subsec:competing-boost}

Near the child cut,

\begin{equation}
  v\prod_{j=1}^{n}\Lambda_j^+(v)
  =(v-\ell_n)+O\!\left((v-\ell_n)^2\right),
  \label{eq:near-cut-weight-plus}
\end{equation}

with the analogous result at $u=-\ell_n$. Thus, sufficiently close to the cut, the recursive modular weight becomes
the usual linear Rindler weight. The local modular flow therefore has the same form as the Bisognano--Wichmann
flow near the horizon \cite{Bisognano:1975ih,Bisognano:1976za,Cardy:2016fqc,Hislop:1981uh}. The iterated logarithms describe the global structure of the recursive flow,
rather than a new local horizon behavior.

Writing

\begin{equation}
  \zeta_{n+1}=A_{n+1}(v)\partial_v+B_{n+1}(u)\partial_u,
  \label{eq:zeta-components}
\end{equation}

the near-cut expansion gives

\begin{align}
  A_{n+1}(v)
  &=g_{n+1}(v-\ell_n)+O\!\left((v-\ell_n)^2\right),
  \label{eq:A-linearization}
  \\
  B_{n+1}(u)
  &=-g_{n+1}(u+\ell_n)+O\!\left((u+\ell_n)^2\right).
  \label{eq:B-linearization}
\end{align}

Hence the vector field vanishes at the bifurcation point $(v,u)=(\ell_n,-\ell_n)$. With the normalization $\zeta_{n+1}=\partial_{t_{n+1}}$,
the transverse linearized coefficient is

\begin{equation}
  \kappa_{n+1}=g_{n+1},
  \qquad
  T_{n+1}=\frac{\kappa_{n+1}}{2\pi}.
  \label{eq:kappa-temperature}
\end{equation}

This agrees with the conformal-Killing-horizon normalization \cite{Jacobson:1993pf} and provides a local consistency check of the
KMS temperature obtained from modular transport.

For comparison, the same translated wedge also admits the ordinary Minkowski Killing boost

\begin{equation}
  \xi_{n+1}^{\mathrm{iso}}
  =g_{n+1}\left[(v-\ell_n)\partial_v-(u+\ell_n)\partial_u\right].
  \label{eq:ordinary-translated-boost}
\end{equation}

It has the same linearization at the bifurcation point and the same $\kappa_{n+1}$, but its global vector field is different
from $\zeta_{n+1}$. The two flows are associated with different states:

\begin{equation}
  \left.\omega_M\right|_{\A(W_{n+1})}
  \longleftrightarrow\xi_{n+1}^{\mathrm{iso}},
  \qquad
  \left.\omega_n\right|_{\A(W_{n+1})}
  \longleftrightarrow\zeta_{n+1}.
  \label{eq:two-state-flow-pairings}
\end{equation}

Thus, agreement near the cut does not imply agreement of the full modular flows. The difference is global and is
encoded in the distinct conformal weights appearing in the two vector fields.

\subsection{Schwarzian stress profiles}
\label{subsec:Schwarzian-profiles}

The modular theorem identifies each state and its dynamics. A concise local
diagnostic is the stress tensor in a common coordinate frame. If a state is
vacuum with respect to $v_n(v)$, then
\begin{equation}
  \langle T_{vv}(v)\rangle_{\omega_n}
  =-\frac{c_+}{24\pi}\{v_n,v\},
  \qquad
  \{f,v\}:=\frac{f'''}{f'}-\frac32\left(\frac{f''}{f'}\right)^2,
  \label{eq:Schwarzian-transform}
\end{equation}
with the analogous $u$ expression. Iterating the Schwarzian chain rule gives
\begin{equation}
  \langle T_{vv}(v)\rangle_{\omega_n}
  =-\frac{c_+}{48\pi}
  \sum_{j=1}^{n}
  \frac{1}{\left[v\prod_{k=1}^{j-1}\Lambda_k^+(v)\right]^2},
  \label{eq:Schwarzian-plus-sum}
\end{equation}
\begin{equation}
  \langle T_{uu}(u)\rangle_{\omega_n}
  =-\frac{c_-}{48\pi}
  \sum_{j=1}^{n}
  \frac{1}{\left[(-u)\prod_{k=1}^{j-1}\Lambda_k^-(u)\right]^2},
  \label{eq:Schwarzian-minus-sum}
\end{equation}
where the empty product is one. These are standard anomalous conformal
transformations; general stress profiles for states with geometric modular
flow are analyzed in Ref.~\cite{Caminiti:2025hjq}, while
\eqref{eq:Schwarzian-plus-sum}--\eqref{eq:Schwarzian-minus-sum} display the
special recursive ancestry of the present family.

At level one,
\begin{equation}
  \langle T_{vv}\rangle_{\omega_R}
  =-\frac{c_+}{48\pi v^2},
  \qquad
  \langle T_{uu}\rangle_{\omega_R}
  =-\frac{c_-}{48\pi u^2}.
  \label{eq:level-one-stress}
\end{equation}
For a parity-symmetric CFT, $c_+=c_-=c$, and at $t=0$,
\begin{equation}
  \langle T_{tt}\rangle_{\omega_R}
  =-\frac{c}{24\pi x^2}
  =-\frac{\pi c}{6}T_{\mathrm{loc}}(x)^2,
  \qquad
  T_{\mathrm{loc}}(x)=\frac1{2\pi x}.
  \label{eq:Fulling-stress-temperature}
\end{equation}
The sign is negative because $\omega_R$ is the Fulling ground state, whereas
the Minkowski vacuum carries positive thermal energy relative to Rindler
time. This is the familiar negative, horizon-singular Fulling profile
\cite{Davies:1976ei}, analogous to the Boulware state near a
black-hole horizon \cite{Boulware:1974dm}. At level two,
\begin{equation}
  \langle T_{vv}(v)\rangle_{\omega_2}
  =-\frac{c_+}{48\pi v^2}
  \left[1+\frac1{\log^2(g_1v)}\right],
  \label{eq:level-two-stress}
\end{equation}
with the analogous $u$ expression.

The state $\omega_n$ lives on $W_n$, whose outer boundary is
$v=\ell_{n-1}$ or $u=-\ell_{n-1}$. Its newest Schwarzian term is singular
there, while earlier formal poles lie outside its physical domain and retain
the ancestry of the iterated map. The profile therefore demonstrates that the hierarchy is not a passive
coordinate relabelling of one vacuum: the states are vacuum-like in
different parent frames but have distinct stress profiles when expressed
in a common frame.

\section{One-generation shift: thermal parent and Rindler child}
\label{sec:one-generation}

The recursive construction admits an equivalent description obtained by shifting the root by one generation.
The key observation is that the region $W_n$, which is a child region in the original recursive construction, can
be regarded as the full reference Minkowski spacetime for the next step. Under this change of description, the
state $\omega_{n-1}$ restricted to $W_n$ becomes a thermal state on the full reference spacetime, while the next child $W_{n+1}$
becomes the ordinary right Rindler wedge $R_n$.

This shift therefore changes the way the same construction is viewed, without changing the underlying
state--algebraic system. The purpose of this section is to make this correspondence precise and to distinguish the
ordinary Rindler observer flow from the modular flow of the thermal state restricted to the Rindler wedge.

\subsection{The one-generation-shift identity}
\label{subsec:shift-identity}

Let $n\geq1$. We define the state on the full level-$n$ reference algebra by
\begin{equation}
  \tau_n
  :=\left(\left.\omega_{n-1}\right|_{\A(W_n)}\right)\circ\Phi_n.
  \label{eq:thermal-parent-state}
\end{equation}
Thus $\tau_n$ is obtained by transporting only the state restricted to the child $W_n$ back to the reference Minkowski
spacetime. This should be distinguished from transporting the entire state $\omega_{n-1}$
($\omega_{n-1}\circ\Phi_{n-1}=\omega_{\mathrm{ref}}$): the latter construction gives
the reference vacuum, whereas \eqref{eq:thermal-parent-state} gives the thermal state relevant for the one-generation shift.

Let $\vartheta_r^{(n)}$ denote the automorphism group induced by ordinary $t_n$-translations on the reference Minkowski
spacetime.
The algebra isomorphism $\Phi_n$ intertwines
$\vartheta_r^{(n)}$ with the level-$n$ parent-time automorphism group
on $\A(W_n)$; geometrically, the corresponding spacetime flows are related
by the conformal chart $\chi_n$.

Since, by the recursive Bisognano--Wichmann theorem, the state $\omega_{n-1}|_{\A(W_n)}$ is KMS with respect to this
parent-time flow, the transported state $\tau_n$ is KMS with respect to ordinary reference time translations.

\begin{proposition}[Thermal parent obtained by unwrapping a child]
\label{prop:thermal-parent}
The state $\tau_n$ is a $\beta_n$-KMS state for ordinary
reference time translations $\vartheta_r^{(n)}$, with
\begin{equation}
  \beta_n=\frac{2\pi}{g_n}.
  \label{eq:thermal-parent-beta}
\end{equation}
\end{proposition}

\begin{proof}
By Theorem~\ref{thm:recursive-bw} at level $n-1$, $\omega_{n-1}|_{\A(W_n)}$ is $\beta_n$-KMS with respect to the level-$n$ parent-time flow.
By construction, $\Phi_n$ intertwines this flow with the ordinary reference time translations $\vartheta_r^{(n)}$. Since the KMS
condition is invariant under $C^*$-algebra isomorphisms, the transported state $\tau_n$ is $\beta_n$-KMS with respect to
$\vartheta_r^{(n)}$.
\end{proof}
We now consider the next recursive step. In the level-$n$ reference spacetime, the region $W_{n+1}$ is represented
by the ordinary right Rindler wedge
\begin{equation}
  W_{n+1}=\chi_n(R_n),
  \label{eq:v2-next-child-reference-wedge}
\end{equation}
as in \eqref{eq:child-split}. Thus, after the shift, the construction has the following form:
\begin{equation}
  \omega_{n-1}\text{ on } W_n\supset W_{n+1}
  \xleftrightarrow{\ \Phi_n^{-1}\ }
  \tau_n\text{ on } M_{\mathrm{ref}}\supset R_n.
  \label{eq:v2-thermal-parent-description}
\end{equation}

The state on $W_n$ has therefore become a thermal state on the full reference spacetime, while its next recursive
child has become an ordinary Rindler wedge.

Let $B_s^{(n+1)}$ be the
ordinary Rindler boost in the reference plane, defined in \eqref{eq:reference-boost}; by definition,
\begin{equation}
  \alpha_s^{(n+1)}
  =\Phi_n\circ B_s^{(n+1)}\circ\Phi_n^{-1}.
  \label{eq:child-boost-through-Phi}
\end{equation}

The preceding observations can now be stated as an exact identity of state--algebra--flow systems.
\begin{theorem}[One-generation-shift identity]
\label{thm:one-generation-shift}
The crossed state--flow system
\begin{equation}
  \bigl(
    \A(W_{n+1}),
    \omega_{n-1}|_{\A(W_{n+1})},
    \alpha^{(n+1)}
  \bigr)
  \label{eq:crossed-physical-system}
\end{equation}
is isomorphic, through the restriction of $\Phi_n$, to the reference Rindler system
\begin{equation}
  \bigl(
    \A(R_n),
    \tau_n|_{\A(R_n)},
    B^{(n+1)}
  \bigr)
  \label{eq:thermal-reference-system}
\end{equation}
\end{theorem}

\begin{proof}
For $A_i\in\A(W_{n+1})$, define the corresponding reference observables by
\[
  A_i^{\mathrm{ref}}:=\Phi_n^{-1}(A_i).
\]

Since
\[
  \Phi_n:\A(R_n)\longrightarrow\A(W_{n+1})
\]
is the restriction of the quasilocal isomorphism to the child algebra, and since the child flow satisfies \eqref{eq:child-boost-through-Phi}, we
have
\[
  \alpha_s^{(n+1)}(A_i)=\Phi_n\!\left(B_s^{(n+1)}(A_i^{\mathrm{ref}})\right).
\]

Using the definition of $\tau_n$ and the fact that $\Phi_n$ preserves products, for arbitrary real $s_1,\ldots,s_m$,
\begin{align}
  &\omega_{n-1}\!\left(
    \alpha_{s_1}^{(n+1)}(A_1)\cdots
    \alpha_{s_m}^{(n+1)}(A_m)
  \right)
  \nonumber\\
  &\qquad=\tau_n\!\left(
    B_{s_1}^{(n+1)}(A_1^{\mathrm{ref}})\cdots
    B_{s_m}^{(n+1)}(A_m^{\mathrm{ref}})
  \right).
  \label{eq:crossed-correlator-identity}
\end{align}

Hence the algebra, state, and observer flow are transported together, and the two descriptions define isomorphic
state--flow systems.
\end{proof}

For $n=1$, the one-generation shift gives the correspondence between the description in terms of the Minkowski
vacuum and Rindler--Rindler observables and the equivalent description in terms of a thermal Minkowski state
and ordinary Rindler observables. Thus Theorem~\ref{thm:one-generation-shift} provides a state--algebra--flow formulation of the second
observation of Ref.~\cite{Kolekar:2013hra}. In particular, the equivalence is not restricted to a particle-number spectrum: all bounded
correlation functions related by the conformal transport satisfy the exact identity \eqref{eq:crossed-correlator-identity}. In the free quasi-free
representation, the corresponding Bogoliubov transformation can be applied to suitably normalized wave-packet
observables and reproduces the occupation-number relation found in Ref.~\cite{Kolekar:2013hra}. Since mode-number operators are unbounded and depend on the choice of
modes and wave-packet normalization, this representation-dependent
statement is not taken as the content of the theorem. The exact statement
is the state--algebra--flow identity \eqref{eq:crossed-correlator-identity}.

\subsection{Thermal-wedge modular flow and observer flow}
\label{subsec:thermal-observer-flows}

There is, however, an important distinction between the observer flow appearing in Theorem~\ref{thm:one-generation-shift} and the
modular flow of the thermal state restricted to the wedge. The reference boost $B_s^{(n+1)}$ is the ordinary Rindler
observer flow. It is generally not the modular automorphism group of the pair

\begin{equation}
  \left(
    \A(R_n),
    \tau_n|_{\A(R_n)}
  \right).
  \label{eq:thermal-wedge-pair}
\end{equation}

The latter modular flow is the thermal-sector geometric flow associated with the restriction of a thermal state
to a Rindler wedge. Its future-directed geometric representative can be obtained by expressing the ordinary
translated-wedge boost in the level-$(n-1)$ coordinates.
The latter is obtained by pulling back the ordinary translated-wedge
Bisognano--Wichmann flow of the ancestor vacuum.

In level-$(n-1)$ coordinates, the wedge corresponding to $R_n$ has tip

\begin{equation}
  (v_{n-1},u_{n-1})=\left(\frac1{g_n},-\frac1{g_n}\right).
  \label{eq:thermal-wedge-tip}
\end{equation}

Its unit-slope translated-wedge boost is
\begin{equation}
  \widetilde\xi_{n+1}^{\mathrm{iso},(n-1)}
  =\left(v_{n-1}-\frac1{g_n}\right)\partial_{v_{n-1}}
  -\left(u_{n-1}+\frac1{g_n}\right)\partial_{u_{n-1}}.
  \label{eq:unit-slope-translated-boost}
\end{equation}

Using the elementary map
\begin{equation}
  v_{n-1}=\frac1{g_n}\e^{g_nv_n},
  \qquad
  u_{n-1}=-\frac1{g_n}\e^{-g_nu_n},
  \label{eq:thermal-map}
\end{equation}

we have
\begin{equation}
  \partial_{v_{n-1}}=\e^{-g_nv_n}\partial_{v_n},
  \qquad
  \partial_{u_{n-1}}=\e^{g_nu_n}\partial_{u_n}.
  \label{eq:thermal-map-derivatives}
\end{equation}

Consequently, the same vector field expressed in level-$n$ coordinates is
\begin{equation}
  \xi_{\beta_n}^{\mathrm{th}}
  =\frac1{g_n}\left[
    \left(1-\e^{-g_nv_n}\right)\partial_{v_n}
    +\left(1-\e^{g_nu_n}\right)\partial_{u_n}
  \right].
  \label{eq:thermal-vector-gn}
\end{equation}
Since $\beta_n=2\pi/g_n$,
\begin{equation}
  \boxed{
  \xi_{\beta_n}^{\mathrm{th}}
  =\frac{\beta_n}{2\pi}\left[
    \left(1-\e^{-2\pi v_n/\beta_n}\right)\partial_{v_n}
    +\left(1-\e^{2\pi u_n/\beta_n}\right)\partial_{u_n}
  \right].}
  \label{eq:BY-vector}
\end{equation}

Near the entangling point this reduces to the unit-slope boost
\[
  v_n\partial_{v_n}-u_n\partial_{u_n},
\]
while deep in the wedge it approaches
\[
  \frac{\beta_n}{2\pi}\partial_{t_n}.
\]
This is the Borchers--Yngvason thermal-sector geometric flow \cite{Borchers:1998ye,Caminiti:2025hjq}. Equation~\eqref{eq:BY-vector} gives its future-directed
physical orientation; the corresponding Tomita modular parameter has the opposite orientation and the $2\pi$
rescaling stated in Appendix~\ref{subsec:tomita-physical-time}.

The two flows appearing in this discussion should therefore be kept distinct. The ordinary reference-plane
Rindler boost is transported directly to the recursive child:

\begin{equation}
  (\chi_n)_*\!\left[
    g_{n+1}\left(v_n\partial_{v_n}-u_n\partial_{u_n}\right)
  \right]
  =\zeta_{n+1}.
  \label{eq:chi-pushforward-Rindler}
\end{equation}

By contrast, the thermal modular flow is related to the ordinary translated-wedge boost of the ancestor description.
Using the elementary map $E_n$,

\begin{equation}
  (E_n)_*\!\left[g_{n+1}\xi_{\beta_n}^{\mathrm{th}}\right]
  =\xi_{n+1}^{\mathrm{iso},(n-1)},
  \label{eq:En-pushforward-thermal}
\end{equation}

where

\begin{equation}
  \xi_{n+1}^{\mathrm{iso},(n-1)}
  =g_{n+1}\left[
    \left(v_{n-1}-\frac1{g_n}\right)\partial_{v_{n-1}}
    -\left(u_{n-1}+\frac1{g_n}\right)\partial_{u_{n-1}}
  \right].
  \label{eq:ancestor-translated-boost}
\end{equation}

Applying $(\chi_{n-1})_*$ subsequently expresses this translated-wedge boost in the original level-zero coordinates.

At the algebraic level, the first relation is simply

\begin{equation}
  \alpha_s^{(n+1)}
  =\Phi_n\circ B_s^{(n+1)}\circ\Phi_n^{-1}.
  \label{eq:observer-flow-algebraic}
\end{equation}

Here $B_s^{(n+1)}$ is the algebraic automorphism induced by the geometric reference-plane boost $B_s^{(n+1)}$.

For the modular groups, let $\Theta_n^{\mathrm{th}}$ denote the normal GNS
extension induced by $\Phi_n$ for the states $\tau_n$ and
$\omega_{n-1}|_{\A(W_n)}$, as in Appendix~\ref{subsec:GNS-transport}. Then

\begin{equation}
  \Theta_n^{\mathrm{th}}
  \circ\sigma_\tau^{\,\tau_n,R_n}
  \circ(\Theta_n^{\mathrm{th}})^{-1}
  =\sigma_\tau^{\,\omega_{n-1},W_{n+1}}.
  \label{eq:modular-group-transport-thermal}
\end{equation}

The resulting dictionary is therefore

\begin{equation}
\begin{gathered}
\text{Rindler observer flow in the reference plane}
\xrightarrow{\ \Phi_n\ }
\text{recursive child observer flow},\\[1mm]
\text{thermal-wedge modular group}
\xrightarrow{\ \Theta_n^{\mathrm{th}}\ }
\text{ancestor-vacuum translated-wedge modular group}.
\end{gathered}
\label{eq:one-generation-dictionary}
\end{equation}

The first row is the one-generation-shift identity and is the state--algebra--flow version of the Kolekar--Padmanabhan comparison. The second row describes the modular structure of the thermal state appearing after
the shift. The two flows need not coincide. In particular, the fact that $\tau_n$ is thermal with respect to ordinary
reference-time translations does not imply that its restriction to $R_n$ is modular with respect to the ordinary
Rindler boost. This distinction is essential for keeping the observer-flow statement of Theorem~\ref{thm:one-generation-shift} separate from
the modular-flow statement.

Therefore
the crossed state $\omega_{n-1}|_{\A(W_{n+1})}$ need not be KMS for the
recursive observer flow $\zeta_{n+1}$ even though all of its bounded
Rindler--Rindler observables admit the thermal-parent interpretation of
Theorem~\ref{thm:one-generation-shift}.

\section{Conclusions and outlook}
\label{sec:conclusions_outlook}

The aim of this work was to formulate recursive Rindler thermality directly at the level of quantum states, local algebras, and modular flow. The main result is the recursive KMS statement
\begin{equation}
    \left.\omega_n\right|_{\mathcal{A}(W_{n+1})}
    \text{ is KMS for } \alpha_s^{(n+1)}
    \quad\text{at}\quad
    \beta_{n+1}=\frac{2\pi}{g_{n+1}}.
    \label{eq:conclusion_recursive_kms}
\end{equation}
The mechanism is the same at every generation: the vacuum restricted to the standard reference wedge is KMS by the Bisognano--Wichmann property, and this state--flow pairing is transported to the recursive child. When the reference vacuum representation is irreducible, the transported parent state is pure, while its restriction to the child is KMS. With wedge duality, the sibling algebra is the modular commutant of the right child. 

For the modular generators, we derived the recursive iterated-log weights, which give the one-sided stress-tensor representation of the modular generator. Near every child cut these weights reduce to the usual linear Rindler form. The Schwarzian stress profiles demonstrate that the hierarchy is not a passive coordinate relabelling of one vacuum: the states are vacuum-like in different parent frames but have distinct stress profiles when expressed in a common frame.

A further result follows by shifting the root of the construction by one generation. The state $\omega_{n-1}$ restricted to $W_n$, when transported back to the level-$n$ reference plane, becomes a thermal state on the full reference spacetime. The next recursive step is then an ordinary Rindler split of this thermal parent. This is an exact state--algebra--flow correspondence for all bounded correlation functions related by the conformal transport.

A natural direction for future work is to extend the recursive
construction to nested causal diamonds, the domains of dependence
of finite spatial intervals. The vacuum modular flows for these
regions are known~\cite{Hislop:1981uh,Cardy:2016fqc,Caminiti:2025hjq},
providing the reference dynamics for an analogous conformal
transport of states and modular flows. It would be useful to
construct explicit recursive families and determine their modular
generators, together with the conditions under which their modular
commutants admit a geometric interpretation in terms of
complementary regions.

\section*{Acknowledgments}
ChatGPT Sol 5.6 and Claude Fable 5 were consulted through the work and used to prepare the manuscript.

\appendix

\section{KMS conventions and modular transport}
\label{app:KMS-conventions}

\subsection{Lower- and upper-strip formulations}
\label{subsec:KMS-strips}

Let $\mathfrak A$ be a $C^*$- or von Neumann algebra, let $\alpha_s$ be a
strongly continuous automorphism group, and let $\omega$ be a state. We use
active Heisenberg evolution
\begin{equation}
  \alpha_s(A)=\e^{\ii sH}A\e^{-\ii sH}
  \label{eq:active-Heisenberg}
\end{equation}
when a Hamiltonian exists. In the lower-strip formulation, $\omega$ is
$\beta$-KMS if, for analytic elements $A,B$, there is a bounded function
$F_{A,B}(z)$ analytic for
\begin{equation}
  -\beta<\operatorname{Im}z<0
  \label{eq:KMS-lower-strip}
\end{equation}
and continuous on the closed strip, with boundary values
\begin{equation}
  F_{A,B}(t)=\omega(\alpha_t(A)B),
  \qquad
  F_{A,B}(t-\ii\beta)=\omega(B\alpha_t(A)).
  \label{eq:KMS-lower-boundaries}
\end{equation}
In finite dimension, for $\mathfrak A=B(\Hh)$ with inner dynamics, this
gives
\begin{equation}
  \rho=\frac{\e^{-\beta H}}{\Tr(\e^{-\beta H})}.
  \label{eq:Gibbs-state}
\end{equation}
The equivalent upper-strip convention uses a bounded $G_{A,B}(z)$ analytic
for $0<\operatorname{Im}z<\beta$ with
\begin{equation}
  G_{A,B}(t)=\omega(A\alpha_t(B)),
  \qquad
  G_{A,B}(t+\ii\beta)=\omega(\alpha_t(B)A).
  \label{eq:KMS-upper-boundaries}
\end{equation}
The two forms are related by stationarity and reversal of complex time. The
strip, operator ordering, $\ii0$ prescription, and sign of the time parameter
must be kept together \cite{Haag:1967sg,Summers:2003tf}.

\subsection{Tomita flow and physical boost time}
\label{subsec:tomita-physical-time}

Let $N\subset B(\Hh)$ and let $\Omega$ be cyclic and separating \cite{Summers:2003tf}. The closure
of
\begin{equation}
  S_0A\Omega=A^*\Omega
  \label{eq:Tomita-S0}
\end{equation}
has polar decomposition
\begin{equation}
  S=J\Delta^{1/2}.
  \label{eq:Tomita-polar}
\end{equation}
We use the standard Tomita modular automorphism group
\begin{equation}
  \sigma_\tau^\Omega(A)=\Delta^{\ii\tau}A\Delta^{-\ii\tau}.
  \label{eq:Tomita-group}
\end{equation}
Tomita--Takesaki theory gives
\begin{equation}
  \sigma_\tau^\Omega(N)=N,
  \qquad
  JNJ=N'.
  \label{eq:Tomita-results}
\end{equation}
For analytic $A,B$, the vector state is KMS at unit modular inverse
temperature. In the upper-strip convention, there is a bounded function
$H_{A,B}(z)$ analytic for $0<\operatorname{Im}z<1$ with
\begin{equation}
  H_{A,B}(t)=\omega_\Omega(\sigma_t^\Omega(A)B),
  \qquad
  H_{A,B}(t+\ii)=\omega_\Omega(B\sigma_t^\Omega(A)).
  \label{eq:Tomita-KMS}
\end{equation}
For a right wedge whose future-directed physical boost parameter is $s$ and
whose normalization is $\kappa>0$, the Bisognano--Wichmann identification in
these conventions is
\begin{equation}
  \sigma_\tau^\Omega=\alpha_{-2\pi\tau/\kappa},
  \qquad
  \alpha_s=\sigma_{-\kappa s/(2\pi)}^\Omega.
  \label{eq:Tomita-physical-relation}
\end{equation}
The minus sign relates the standard Tomita parameter to the chosen active
physical time orientation. The corresponding physical inverse temperature
is
\begin{equation}
  \beta_s=\frac{2\pi}{\kappa}.
  \label{eq:physical-beta}
\end{equation}
Writing $L=-\log\Delta$, the canonical standard implementer of the physical
flow has generator
\begin{equation}
  H_{\mathrm{std}}=\frac{\kappa}{2\pi}L.
  \label{eq:Hstd}
\end{equation}
A one-sided stress-tensor generator can implement the same derivation on the
right algebra while differing from $H_{\mathrm{std}}$ by an operator
affiliated with the commutant, as in \eqref{eq:L-KR-KL}.

\subsection{Transport through GNS representations}
\label{subsec:GNS-transport}

Let $\Phi:\mathfrak A\to\mathfrak B$ be a $*$-isomorphism, let $\omega$ be a
state on $\mathfrak A$, and set $\psi=\omega\circ\Phi^{-1}$. The map
\begin{equation}
  V\pi_\omega(A)\Omega_\omega
  =\pi_\psi(\Phi(A))\Omega_\psi
  \label{eq:GNS-unitary}
\end{equation}
preserves inner products and extends to a unitary. For a local subalgebra
$\mathfrak m\subset\mathfrak A$, set
\begin{equation}
  N=\overline{\pi_\omega(\mathfrak m)}^{\,\mathrm{wot}},
  \qquad
  P=\overline{\pi_\psi(\Phi(\mathfrak m))}^{\,\mathrm{wot}}.
  \label{eq:GNS-local-closures}
\end{equation}
Then $VNV^*=P$. If the GNS vectors are cyclic and separating,
\begin{equation}
  S_P=VS_NV^*,
  \qquad
  \Delta_P=V\Delta_NV^*,
  \qquad
  J_P=VJ_NV^*.
  \label{eq:transported-modular-data}
\end{equation}
At the $C^*$ level, if $\alpha_s^\Phi=\Phi\alpha_s\Phi^{-1}$, then for
entire analytic elements
\begin{equation}
  \psi\!\left(\alpha_z^\Phi(\Phi(A))\Phi(B)\right)
  =\omega(\alpha_z(A)B),
  \label{eq:KMS-isomorphism-transport}
\end{equation}
so the strip, ordering, and inverse temperature transport unchanged.

\section{Conformal-net assumptions and local transport}
\label{app:conformal-net}

\subsection{Local net and covariance}
\label{subsec:local-net-covariance}

A two-dimensional conformal net assigns a von Neumann algebra $\A(O)$ to
each bounded causal diamond (also called a double cone)
\begin{equation}
  O=I_+\times I_-=I^+(p)\cap I^-(q),
  \label{eq:causal-diamond-definition}
\end{equation}
and satisfies isotony, locality, positive energy, covariance, and the vacuum
condition. We assume bosonic locality; graded field nets require graded
locality and twisted duality. For brevity in this appendix only, we write
$O\Subset R$ when $\overline O$ is compact and $\overline O\subset R$.

Diffeomorphism covariance supplies isomorphisms for orientation-preserving
chiral maps on bounded intervals. If
\begin{equation}
  f_+:I_+\to J_+,
  \qquad
  f_-:I_-\to J_-,
  \label{eq:chiral-maps}
\end{equation}
then there is a local isomorphism
\begin{equation}
  \alpha_{f_+,f_-}:\A(I_+\times I_-)
  \longrightarrow\A(J_+\times J_-).
  \label{eq:local-conformal-isomorphism}
\end{equation}
The exponential maps used here are not single vacuum-preserving global
diffeomorphisms of the compactified line. On every compact interval away
from a horizon endpoint they can, however, be extended to circle
diffeomorphisms. We assume the induced local maps are chosen coherently under
inclusion; this yields the norm-inductive-limit isomorphism
\eqref{eq:Phi-child}.

Local normality follows from the same implementation. Fix $O\Subset W_n$,
set $O^{\mathrm{ref}}=\chi_n^{-1}(O)$, and choose global chiral
circle-diffeomorphism extensions $f_\pm$ of the chart on $O^{\mathrm{ref}}$. In
the original vacuum representation let $U(f_+,f_-)$ implement these
extensions. If $A\in\A(O^{\mathrm{ref}})$ and
\begin{equation}
  B=\Phi_n(A)=U(f_+,f_-)AU(f_+,f_-)^*\in\A(O),
  \label{eq:local-implementation}
\end{equation}
then
\begin{equation}
  \omega_n(B)
  =\langle\Omega_0,A\Omega_0\rangle
  =\langle U(f_+,f_-)\Omega_0,
    B\,U(f_+,f_-)\Omega_0\rangle.
  \label{eq:local-normality}
\end{equation}
Thus $\omega_n|_{\A(O)}$ is a vector state in the original vacuum
representation and is normal.

\subsection{Bisognano--Wichmann property and duality}
\label{subsec:BW-duality}

We assume the reference vacuum vector is cyclic and separating for the
reference wedge algebras and that
\begin{equation}
  (\Delta_R^{\mathrm{ref}})^{\ii\tau}
  =U^{\mathrm{ref}}\!\left(B_R(-2\pi\tau)\right),
  \label{eq:reference-BW-property}
\end{equation}
where $B_R(r)$ is the right-wedge-preserving Lorentz boost of rapidity $r$
and $U^{\mathrm{ref}}$ is its reference vacuum representation. This is the
Bisognano--Wichmann property in the sign convention of
\eqref{eq:Tomita-group}. For the geometric purifier we also assume essential
wedge duality,
\begin{equation}
  (M_R^{\mathrm{ref}})'=M_L^{\mathrm{ref}}.
  \label{eq:essential-wedge-duality}
\end{equation}
Both statements are standard in the relevant conformal-net setting
\cite{Brunetti:1992zf}. If only the first holds, the recursive KMS
theorem remains intact but the commutant cannot automatically be identified
with the algebra of $D_n$.

\bibliographystyle{unsrt}
\bibliography{main}

\end{document}